\documentclass{entcs} 

\usepackage[utf8]{inputenc}

\usepackage{entcsmacro} 
\usepackage{array}
\usepackage{semantic,stmaryrd,tikz,multirow}

\usepackage{listings}
\usepackage{lstjava}
\newcolumntype{M}[1]{>{$}m{#1}<{$}}
\newcolumntype{C}[1]{>{$}c<{$}}
\newcolumntype{L}[1]{>{$}l<{$}}
\newcolumntype{R}[1]{>{$}r<{$}}

\newcommand{\clinit}{\ensuremath{\texttt{<clinit>}}}

\newcommand{\may}{\ensuremath{\mathit{May}}}
\newcommand{\must}{\ensuremath{\mathit{Must}}}
\newcommand{\wf}{\ensuremath{\mathit{Wf}}}

\newcommand{\staticC}{\texttt{\dots{}C\dots}}
\newcommand{\flow}{\mathit{flow}}
\newcommand{\intraflow}{\flow_{\rm intra}}
\newcommand{\interflow}{\flow_{\rm inter}}
\newcommand{\clinitflow}{\flow_{\rm clinit}}

\def\lastname{Hubert and Pichardie}

\begin{document}

\begin{frontmatter}
  \title{Soundly Handling Static Fields:\\
    Issues, Semantics and Analysis}

  \author[CnrsIrisa]{Laurent
    Hubert\thanksref{email}\thanksref{regionbretagne}}
  \author[InriaIrisa]{David Pichardie\thanksref{email}}

  \address[CnrsIrisa]{CNRS, IRISA, Campus Beaulieu , F-35042 Rennes
    Cedex, France}%
  \address[InriaIrisa]{INRIA, Centre Rennes --- Bretagne Atlantique\\
    IRISA, Campus Beaulieu , F-35042 Rennes Cedex, France}

  \thanks[email]{Email:\texttt{first.last@irisa.fr}}
  \thanks[regionbretagne]{This work was supported in part by the
    Région Bretagne}

  \begin{abstract}
    %
    Although in most cases class initialization works as expected,
    some static fields may be read before being initialized, despite
    being initialized in their corresponding class initializer.
    We propose an analysis which compute, for each program point, the
    set of static fields that must have been initialized and discuss
    its soundness.
    We show that such an analysis can be directly applied to identify
    the static fields that may be read before being initialized and to
    improve the precision while preserving the soundness of a
    null-pointer analysis.
  \end{abstract}

\end{frontmatter}

\section{Introduction}
\label{sec:introduction}

Program analyses often rely on the data manipulated by programs and
can therefore depend on their static fields.  Unlike instance fields,
static fields are unique to each class and one would like to benefit
from this uniqueness to infer precise information about their content.

When reading a variable, be it a local variable or a field, being sure
it has been initialized beforehand is a nice property.  Although the
Java bytecode ensures this property for local variables, it is not
ensured for static and instance fields which have default values.

Instance fields and static fields are not initialized the same way:
instance fields are usually initialized in a constructor which is
explicitly called whereas static fields are initialized in class
initializers which are implicitly and lazily invoked.  This makes the
control flow graph much less intuitive.

The contributions of this work are the followings.
\begin{itemize}
\item We recall that implicit lazy static field initialization makes
  the control flow graph hard to compute.
\item We identify some code examples that would need to be ruled out
  and some other examples that would need \emph{not} to be ruled out.
\item We propose a language to study the initialization of static
  fields.
\item We propose a formal analysis to infer an under-approximation of
  the set of static fields that have already been initialized for each
  program point.
\item We propose two possible applications for this analysis: a direct
  application is to identify potential bugs and another one is to
  improve the precision while keeping the correctness of a
  null-pointer analysis.
\end{itemize}

The rest of this paper is organized as follows.  We recall in
Sect.~\ref{sec:why-difficult} that the actual control flow graph which
includes the calls the class initializers it not intuitive and give
some examples.  We present in Sect.~\ref{sec:language} the syntax and
semantics of the language we have chosen to formalize our analysis.
Section~\ref{sec:analysis} then presents the analysis, first giving an
informal description and then its formal definition.  We then explain
in Sect.~\ref{sec:extensions} how the analysis can be extended to
handle other features of the Java bytecode language.  In
Sect.~\ref{sec:two-possible-uses}, we give two possible applications
of this analysis.  Finally, we discuss the related work in
Sect.~\ref{sec:related-work} and conclude in
Sect.~\ref{sec:conclusion}.

\section{Why Static Analysis of Static Fields is Difficult?}
\label{sec:why-difficult}

The analysis we herein present works at the bytecode level but, for
sake of simplicity, code examples are given in Java.  As this paper is
focused on static fields, all fields are assumed to be static unless
otherwise stated.

In Java, a field declaration may include its initial value, such as
\texttt{A.f} in Fig.~\ref{fig:depend_on_main}.  A field can also be
initialized in a special method called a \emph{class initializer},
which is identified in the Java source code with the \texttt{static}
keyword followed by no signature and a method body, such as in class
\texttt{B} in the same figure.
If a field is initialized with a compile-time constant expression, the
compiler from Java to bytecode may translate the initialization into a
\emph{field initializer} (cf.~\cite{lindholm99:jvm_spec},
Sect.~4.7.2), which is an attribute of the field.  At run time, the
field should be set to this value before running the class
initializer.
In-line initializations that have not been compiled as field
initializers are prepended in textual order to the class initializer,
named \clinit{} at the bytecode level.
For this analysis, we do not consider field initializers but focus on
class initializers as they introduce the main challenges.  Although
this simplification is sound, it is less precise and we explain how to
extend our analysis to handle field initializers in
Sect.~\ref{sec:initialization-order}.

The class initialization process is not explicitly handled by the
user: it is forbidden to explicitly called a \clinit{} method.
Instead, every access (read or write) to a field of a particular class
or the creation of an instance of that same class requires that the
JVM (Java Virtual Machine) has \emph{invoked} the class initializer
of that class.  This method can contain arbitrary code and may trigger
the initialization of other classes and so on.

The JVM specification~\cite{lindholm99:jvm_spec} requires class
initializers to be invoked \emph{lazily}.
This implies that the order in which classes are initialized depends
on the \emph{execution path}, it is therefore not decidable in
general.

\begin{figure}
  \begin{center}
\begin{lstlisting}
class A extends Object{static B f = new B();}
class B extends Object{
  static B g;
  static {
    g = A.f;
  }
}
\end{lstlisting}
  \end{center}
  \caption{Initial values can depend on foreign code: in this example,
    the main program should first use \texttt{B} for the
    initialization to start from \texttt{B} to avoid \texttt{B.g} to
    be \texttt{null}.}
  \label{fig:depend_on_main}
\end{figure}

The JVM specification also requires class initializers to be invoked
at most once.
This avoids infinite recursions in the case of circular dependencies
between classes, but it also implies that when reading a field it may
not contain yet its ``initial'' value.  For example, in
Fig.~\ref{fig:depend_on_main}, the class initializer of \texttt{A}
creates an instance of \texttt{B} and therefore requires that the
class initializer of \texttt{B} has been invoked.  The class
initializer of \texttt{B} reads a field of \texttt{A} and therefore
requires that the class initializer of \texttt{A} has been invoked.
\begin{itemize}
\item If \texttt{B.\clinit} is invoked before \texttt{A.\clinit}, then
  the read access to the field \texttt{A.f} triggers the invocation of
  \texttt{A.\clinit}.  Then, as \texttt{B.\clinit} has already been
  invoked, \texttt{A.\clinit} carries on normally and creates an
  instance of class \texttt{B}, store its reference to the field
  \texttt{A.f} and returns.  Back in \texttt{B.\clinit}, the field
  \texttt{A.f} is read and the reference to the new object is also
  affected to \texttt{B.g}.
\item If \texttt{A.\clinit} is invoked before \texttt{B.\clinit}, then
  before allocating a new instance of \texttt{B}, the JVM has to
  initialize the class \texttt{B} by calling \texttt{B.\clinit}.  In
  \texttt{B.\clinit}, the read access to \texttt{A.f} does not trigger
  the initializer of \texttt{A} because \texttt{A.\clinit} has already
  been started.  \texttt{B.\clinit} then reads \texttt{A.f}, which has
  not been initialized yet, \texttt{B.g} is therefore set to the
  default value of \texttt{A.f} which is the \texttt{null} constant.
\end{itemize}
This example shows that the order in which classes are initialized
modifies the semantics.
The issue shown in Fig.~\ref{fig:depend_on_main} is not limited to
reference fields.  In the example in
Fig.~\ref{fig:integer_depend_on_main}, depending on the initialization
order, \texttt{A.CST} will be either 0 or 5, while \texttt{B.SIZE}
will always be 5.
\begin{figure}
  \centering
\begin{lstlisting}
class A extends Object{
  public static int CST= B.SIZE;}
class B extends Object{
  public static int SIZE = A.CST+5;
}
\end{lstlisting}
  \caption{Integer initial values can also depend on foreign code}
  \label{fig:integer_depend_on_main}
\end{figure}


One could notice that those problems are related to the notion of
circular dependencies between classes and may think that circular
dependencies should be avoided.
Figure~\ref{fig:one-class} shows an example with a single class.  In
(a), \texttt{A.ALL} is read in the constructor before it has been
initialized and it leads to a \texttt{NullPointerException}. (b) is
the correct version, where the initializations of \texttt{ALL} and
\texttt{EMPTY} have been switched.  This example is an extract of
\texttt{java.lang.Character\$UnicodeBlock} of Sun's Java Runtime
Environment (JRE) that we have simplified: we want the analysis to
handle such cases.
If we consider that \texttt{A} depend on itself, then we forbid way to
much programs.  If we do not consider that \texttt{A} depend on itself
and we do not reject (a) then the analysis is incorrect.  We therefore
cannot rely on circular dependencies between classes.

\begin{figure}\normalsize
  \begin{minipage}{.49\linewidth}
\begin{lstlisting}
class A extends Object{
  static A EMPTY=new A("");
  static HashMap ALL
            =new HashMap();

  String name;
  public A(String name){
    this.name = name;
    ALL.add(name,this);
  }
}
\end{lstlisting}
    \begin{center} \footnotesize (a) An uninitialized field read
      \\leads to a \texttt{NullPointerException}
    \end{center}
  \end{minipage}
  \begin{minipage}{.48\linewidth}
\begin{lstlisting}
class A extends Object{
  static HashMap ALL
            =new HashMap();
  static A EMPTY=new A("");

  String name;
  public A(String name){
    this.name = name;
    ALL.add(name,this);
  }
}
\end{lstlisting}
    \begin{center} \footnotesize (b) No uninitialized field is read
    \end{center}
  \end{minipage}

  \caption{The issue can arise with a single class}
  \label{fig:one-class}
\end{figure}

\section{The Language}
\label{sec:language}
In this section we present the program model we consider in this work.
This model is a high level description of the bytecode program that
discards the features that are not relevant to the specific problem of
static field initialization.

\subsection{Syntax}
\label{sec:syntax}

\newcommand{\ppt}{\mathbb{P}}
\newcommand{\first}{\mathit{first}}
\newcommand{\last}{\mathit{last}}
\newcommand{\fields}{\mathbb{F}}
\newcommand{\class}{\mathbb{C}}
\newcommand{\method}{\mathbb{M}}
\newcommand{\instrAt}{\mathit{instr}}
\newcommand{\any}{\mathtt{any}}
\newcommand{\getS}{\mathtt{getstatic}}
\newcommand{\putS}{\mathtt{put}}
\newcommand{\invoke}{\mathtt{invoke}}
\newcommand{\return}{\mathtt{return}}
\newcommand{\pto}{\rightharpoonup}

We assume a set $\ppt$ of program points, a set $\fields$ of field
names, a set $\class$ of class names and a set $\method$ of method
names.  For each method $m$, we note $m.\first$ the first program point
of the method $m$.  For convenience, we associate to each method a
distinct program point $m.\last$ which models the output point of the
method.  For each class $C$ we note $C.\clinit$ the name of the class
initializer of $C$.  We only consider four kinds of instructions.
\begin{itemize}
\item $\putS(f)$ updates a field $f\in\fields$.
\item $\invoke$ calls a method (we do not mention the name of the
  target method because it would be redundant with the $\interflow$
  information described below).
\item $\return$ returns from a method.
\item $\any$ models any other intra-procedural instruction that does
  not affect static fields.
\end{itemize}
As the semantics presented below will demonstrate, any instruction of
the standard sequential Java bytecode can be represented by one of
these instructions.

The program model we consider is based on control flow relations that
must have been computed by some standard control flow analysis.
\begin{definition}\sloppy
  A program is a 5-tuple
  $(m_0,\,\instrAt,\,\intraflow,\,\interflow,\,\clinitflow)$ where:
  \begin{itemize}
  \item $m_0\in\method$ is the method where the program execution
    starts;
  \item $\instrAt\in\ppt\pto\{\putS(f),~\invoke,\return,\any\}$ is a
    partial function that associates program points to instructions;
  \item $\intraflow\subseteq \ppt\times\ppt$ is the set of
    intra-procedural edges;
  \item $\interflow\subseteq \ppt\times\method$ is the set of
    inter-procedural edges, which can capture dynamic method calls;
  \item $\clinitflow\in \ppt\pto\class$ is the set of initialization
    edges which forms a partial function since an instruction may only
    refer to one class;
  \end{itemize}
  and such that $\instrAt$ and $\intraflow$ satisfy the following
  property:
  \begin{quote}
    For any method $m$, for any program point $l\in\ppt$ that is
    reachable from $m.\first$ in the intra-procedural graph given by
    $\intraflow$, and such that $\instrAt(l)=\return$, $(l,m.\last)$
    belongs to $\intraflow$.
  \end{quote}
\end{definition}
In practice, $\clinitflow$ will contain all the pairs $(l,C)$ of a
bytecode program such that the instruction found at program point $l$
is of the form $\mathtt{new}~C$, $\mathtt{putstatic}~C.f$,
$\mathtt{getstatic}~C.f$ or $\mathtt{invokestatic}~C.m$ (see
Section~2.17.4 of \cite{lindholm99:jvm_spec}).\footnote{%
  To be completely correct we also need to add an edge from the
  beginning of the \texttt{static void main} method to the initializer
  of its class.  We can also handle correctly superclass and interface
  initialization without deep modification of the current
  formalization.}



Figure~\ref{fig:example} in Sect.~\ref{sec:formal-specification}
presents an example of program with its three control flow
relations.  In this program, the main method $m_0$ contains two
distinct paths that lead to the call of a method $m$.  In the first
one, the class $\mathtt{A}$ is initialized first and its initializer
triggers the initialization of $\mathtt{B}$.  In the second path,
$\mathtt{A}$ is not initialized but $\mathtt{B}$ is.
$\mathtt{A}.\clinit$ is potentially called at exit point $8$ but since
$8$ is only reachable after a first call to $\mathtt{A}.\clinit$, this
initialization edge is never taken.

\subsection{Semantics}
\label{sec:semantics}

\newcommand{\Pow}{\mathcal{P}}
\newcommand{\eqdef}{\,\stackrel{\textrm{\scriptsize def}}{=}\,}
\newcommand{\Ref}{\mathrm{Location}}
\newcommand{\Value}{\mathrm{Value}}
\newcommand{\Local}{\mathrm{Local}}
\newcommand{\OpStack}{\mathrm{OpStack}}
\newcommand{\Dynamics}{\mathrm{Dynamics}}
\newcommand{\Heap}{\mathrm{Heap}}
\newcommand{\Static}{\mathrm{Static}}
\newcommand{\Frame}{\mathrm{Frame}}
\newcommand{\State}{\mathrm{State}}
\newcommand{\IntraState}{\State_{\mathrm{intra}}}
\newcommand{\FinalState}{\State_{\mathrm{final}}}
\newcommand{\History}{\mathrm{History}}
\newcommand{\CallStack}{\mathrm{Callstack}}
\newcommand{\st}[1]{<\!\!#1\!\!>}
\newcommand{\fr}{\mathit{fr}} \newcommand{\cs}{\mathit{cs}}
\newcommand{\hist}{h}
\newcommand{\stt}{\mathit{st}} \newcommand{\os}{\mathit{os}}
\newcommand{\cons}{\!::\!}  \newcommand{\undef}{\Omega}
\newcommand{\nil}{\varepsilon} \newcommand{\marked}[1]{\llparenthesis
  #1 \rrparenthesis}

The analysis we consider does not take into account the content of
heaps, local variables or operand stacks.  To simplify the
presentation, we hence choose an abstract semantics which is a
conservative abstraction of the concrete standard semantics and does
not explicitly manipulate these domains.

The abstract domains the semantics manipulates are presented below.
$$
\begin{array}{rclcl}
  v    &\in& \Value &      & (\mathit{abstract}) \\
  s    &\in& \Static&\eqdef& \fields \to \Value+\{\undef\}\\
  \hist  &\in& \History&\eqdef& \Pow(\class) \\
  \cs &\in& \CallStack &\eqdef& (\{0,1\}\times\ppt)^\star\\
  \st{l,\cs,s,\hist}&\in& \State & \eqdef & \ppt\times\CallStack\times\Static\times\History\\ 
\end{array}
$$
A field contains either a value or a default value represented by the
symbol $\undef$.  Since a class initializer cannot be called twice in a
same execution, we need to remember the set of classes whose
initialization has been started (but not necessarily ended).  This
is the purpose of the element $h\in\History$.  Our language is given a
small-step operational semantics with states of the form
$\st{l,\cs,s,\hist}$, where the label $l$ uniquely identifies the
current program point, $\cs$ is a call stack that keeps track of the
program points where method calls have occurred, $s$ associates to
each field its value or $\undef$ and $\hist$ is the history of class
initializer calls.  Each program point $l$ of a call stack is tagged
with a Boolean which indicates whether the call in $l$ was a call to a
class initializer (the element of the stack is then noted
$\marked{l}$) or was a standard method call (simply noted ${l}$).

\newcommand{\Reach}[2]{\mathit{Reach}(#1,#2)}
\newcommand{\stepi}{\,\to_1\,} \newcommand{\stepii}{\,\to_2\,}
\newcommand{\step}{\,\to\,} \newcommand{\dom}{\operatorname{dom}}
\newcommand{\stepstar}{\,\to^\star\,}
\newcommand{\NeedInit}{\mathit{NeedInit}}
\newcommand{\ia}[1]{\llbracket{#1}\rrbracket}
\newcommand{\init}{\mathrm{init}} \newcommand{\inA}{\mathrm{in}}
\newcommand{\callA}{\mathrm{call}} \newcommand{\outA}{\mathrm{out}}

The small-step relation $\step \subseteq \State\times\State$ is given
in Fig.~\ref{fig:sem} (we left implicit the program
$(m_0,\instrAt,\intraflow,\interflow,\clinitflow)$ that we consider).
It is based on the relation
$\NeedInit\subseteq\ppt\times\class\times\History$ defined by
$$
\NeedInit(l,C,\hist)\quad \eqdef \quad \clinitflow(l)=C~ \land~
C\not\in\hist
$$
which means that the class initializer of class $C$ must be called at
program point $l$ if and only if there is a corresponding edge in
$\clinitflow$ and $C.\clinit$ has not been called yet (\emph{i.e.}
$C\not\in\hist$).

\begin{figure}
$$
\begin{array}[c]{c}
  \inference%
  {\NeedInit(l,C,\hist)}%
  {\st{l,\cs,s,\hist} \step \st{C.\clinit.\first,\marked{l}::\cs,s,\hist\cup\{C\}}}\\[2.5ex]
  \inference%
  {\st{l,\cs,s,\hist} \stepi \st{l',\cs',s',\hist'} & \forall C,\ \neg\NeedInit(l,C,\hist)}%
  {\st{l,\cs,s,\hist} \step \st{l',\cs',s',\hist'}}\\[5ex]
  \inference%
  {\instrAt(l)=\putS(f)  & \intraflow(l,l') & v\in\Value}
  {\st{l,\cs,s,\hist} \stepi \st{l',\cs,s[f\mapsto v],\hist}}\\[2ex]
  \inference%
  {\instrAt(l)=\any  & \intraflow(l,l')}%
  {\st{l,\cs,s,\hist} \stepi \st{l',\cs,s,\hist}}~
  \inference%
  {\instrAt(l)=\invoke & \interflow(l,m)}%
  {\st{l,\cs,s,\hist} \stepi \st{m.\first,l::\cs,s,\hist}}\\[2ex]
  \inference%
  {\instrAt(l)=\return &\!\!\! \intraflow(l',l'')}%
  {\st{l,l'::\cs,s,\hist} \stepi \st{l'',\cs,s,\hist}}~
  \inference%
  {\instrAt(l)=\return &\!\!\! \st{l',\cs,s,\hist} \stepi \stt'}%
  {\st{l,\marked{l'}::\cs,s,\hist} \stepi \stt'}\\[2ex]
\end{array}
$$  
\caption{Operational semantics}
\label{fig:sem}
\end{figure}

The relation $\to$ is defined by two rules.  In the first one, the
class initializer of class $C$ needs to be called.  We hence jump to
the first point of $C.\clinit$, push on the call stack the previous
program point (marked with the flag $\marked{\cdot}$) and record $C$
in the history $\hist$.  In the second rule, there is no need to
initialize a class, hence we simply use the standard semantic of the
current instruction, given by the relation
$\stepi\subseteq\State\times\State$.

The relation $\stepi$ is defined by five rules.
The first one corresponds to a field update $\putS(f)$: an arbitrary
value $v$ is stored in field $f$.
The second rule illustrates that the instruction $\any$ does not affect
the visible elements of the state.
For a method call (third rule), the current point is pushed on the
call stack and the control is transferred to (one of) the target(s) of
the inter-procedural edge.
At last, the instruction $\return$ requires two rules.  In the first
case, a standard method call, the transfer comes back to the
intra-procedural successor of the caller.  In the second case, a class
initializer call, we have finished the initialization
and we must now use the standard semantic $\stepi$ of the pending
instruction in program point $l$.

We end this section with the formal definition of the set of reachable
states during a program execution.  An execution starts in the main
method $m_0$ with an empty call stack, an empty historic and with all
fields associated to the default value $\undef$.

\begin{definition}[Reachable States]
  The set of reachable states of a program
  $p=(m_0,\instrAt,\intraflow,\interflow,\clinitflow)$ is defined by
 $$\ia{p} = \left\{\ \st{i,\cs,s,\hist}\ \mid\ 
   \st{m_0.\first,\nil,\lambda f.\undef,\emptyset} \stepstar
   \st{i,\cs,s,\hist}\ \right\}$$
\end{definition}

\section{A Must-Have-Been-Initialized Data Flow Analysis}
\label{sec:analysis}

In this section we present a sound data flow analysis that allows to
prove a static field has already been initialized at a particular
program point.  We first give an informal presentation of the
analysis, then present its formal definition and we finish with the
statement of a soundness theorem.

\subsection{Informal presentation}
\label{sec:informal-presentation}
For each program point we want to know as precisely as possible, which
fields we are sure we have initialized.  Since fields are generally
initialized in class initializers, we need an inter-procedural
analysis that infers the set of fields $\wf$ that \emph{must} have
been initialized at the end of each method.  Hence at each program
point $l$ where a method is called, be it a class initializer or
another method, we will use this information to improve our knowledge
about initialized fields.

However, in the case of a call to a class initializer, we need to be
sure the class initializer will be effectively executed if we want to
safely use such an information.  Indeed, at execution time, when we
reach a point $l$ with an initialization edge to a class $C$, despite
$\clinitflow(l)=C$, two exclusive cases may happen:
\begin{enumerate}
\item $C.\clinit$ has not been called yet: $C.\clinit$ is immediately
  called.
\item $C.\clinit$ has already been called (and may still be in
  progress): $C.\clinit$ will not be called a second time.
\end{enumerate}
Using the initialization information given by $C.\clinit$ is safe only
in case (i), \emph{i.e.} the control flow information in $\clinitflow$
is not precise enough.  To detect case {(i)}, we keep track in a
flow-sensitive manner of the class initializer that \emph{may} have
been called during all execution reaching a given program point.  We
denote by $\may$ this set.  Here, if $C$ is not in $\may$, we are sure
to be in case {(i)}.
%
$\may$ is computed by gathering, in a flow sensitive way, all classes
that may be initialized starting from the main method.  Implicit calls
to class initializer need to be taken in account, but the smaller
$\may$ is, the better.


\begin{figure}
  \centering
  \begin{minipage}{.8\linewidth}
\begin{lstlisting}[numbers=left]
class A{static int f = 1;}
class B{static int g = A.f;}
class C{
  public static void main(String[] args){
    ... = B.g;
    ... = A.f;
    ... = B.g;
  }
}
\end{lstlisting}
  \end{minipage}
  \caption{Motivating the \must{} set}
  \label{fig:motivating-must-set}
\end{figure}

For the simplicity's sake we consider in this work a
context-insensitive analysis where for each method, all its calling
contexts are merged at its entry point.
Consider the program example given in
Fig.~\ref{fig:motivating-must-set}.  Before line 5, \may{} only
contains $\mathtt{C}$, the class of the \texttt{main} method.  There
is an implicit flow from line 5 to the class initializer of
$\mathtt{B}$.  At the beginning of the class initializer of
$\mathtt{B}$, \may{} equals to $\{\mathtt{B},\mathtt{C}\}$.  We
compute the set of fields initialized by \texttt{A.\clinit}, which is
$\{\mathtt{A.f}\}$.  As $\mathtt{A}$ is not in \may{} at the beginning
of \texttt{B.\clinit}, we can assume the class initializer will be
fully executed before the actual read to \texttt{A.f} occurs, so it is
a safe read.  However, when we carry on line 7, the \may{} set
contains $\mathtt{A}$, $\mathtt{B}$ and $\mathtt{C}$.  If we flow this
information to \texttt{B.\clinit}, then the merged calling context of
\texttt{B.\clinit} is now $\{\mathtt{A},\mathtt{B},\mathtt{C}\}$,
which makes impossible to assume anymore that \texttt{A.\clinit} is
called at line 2.
To avoid such an imprecision, we try to propagate as few calling
context as possible to class initializers by computing in a
flow-sensitive manner a second set of class whose initializer
\emph{must} have already been called in all execution reaching a given
program point.  We denote by $\must$ this set.  Each time we encounter
an initialization edge for a class $\mathtt{C}$, we add $\mathtt{C}$
to $\must$ since $\mathtt{C}.\clinit$ is either called at this point,
or has already been called before.  If $\mathtt{C}\in\must$ before an
initialization edge for $\mathtt{C}$, we are sure $\mathtt{C}.\clinit$
will not be called at this point and we can avoid to propagate a
useless calling context to $\mathtt{C}.\clinit$.

To sum up, our analysis manipulates, in a flow sensitive manner, three
sets $\may$, $\must$ and $\wf$. $\may$ and $\must$ correspond to
a control flow analysis that allows to refine the
initialization graph given by $\clinitflow$.
The more precise control flow graph allows a finer tracking of field
initialization and therefore a more precise $\wf$.

\subsection{Formal specification}
\label{sec:formal-specification}

In this part we consider a given program $p=
(m_0,\instrAt,\intraflow,\interflow,\clinitflow)$.  We note
$\Pow_p(\class)$ (resp. $\Pow_p(\fields)$) the finite set of classes
(reps. fields) that appears in $p$.  For each program point $l\in\ppt$,
we compute before and after the current point three sets of data
$(\may,\must,\wf)\in\Pow_p(\class)\times\Pow_p(\class)\times\Pow_p(\fields)$.
Since $\may$ is a \emph{may} information, and $\must$ and $\wf$ are
\emph{must} information, the underlying lattice of the data flow
analysis is given by the following definition.
\begin{definition}[Analysis lattice]
  The analysis lattice is 
  $(A^\sharp,\sqsubseteq,\sqcup,\sqcap,\bot,\top)$ where:
  \begin{itemize}
  \item $A^\sharp =
    \Pow_p(\class)\times\Pow_p(\class)\times\Pow_p(\fields)$.
  \item $\bot=(\emptyset,\Pow_p(\class),\Pow_p(\fields))$.
  \item $\top=(\Pow_p(\class),\emptyset,\emptyset)$.
  \item for all $(\may_1,\must_1,\wf_1)$ and $(\may_2,\must_2,\wf_2)$
    in $A^\sharp$,\vspace*{-3ex}

    $$
    \begin{array}[c]{l}
      (\may_1,\must_1,\wf_1)\sqsubseteq(\may_2,\must_2,\wf_2)~\mathrm{iff}\\
      ~~~~~~~~~~~~~~~~~~~~
      \may_1 \subseteq \may_2,\ \must_1 \supseteq \must_2~\mathrm{and}~\wf_1 \supseteq \wf_2\\
    \end{array}
    $$\vspace*{-2ex}
    $$
    \begin{array}[c]{l}
      (\may_1,\must_1,\wf_1)\sqcup(\may_2,\must_2,\wf_2) = \\
      ~~~~~~~~~~~~~~~~~~~~
      (\may_1 \cup \may_2, \must_1 \cap \must_2,\wf_1 \cap \wf_2)\\
    \end{array}
    $$\vspace*{-2ex}
    $$
    \begin{array}[c]{l}
      (\may_1,\must_1,\wf_1)\sqcap(\may_2,\must_2,\wf_2) = \\
      ~~~~~~~~~~~~~~~~~~~~
      (\may_1 \cap \may_2, \must_1 \cup \must_2,\wf_1 \cup \wf_2)
    \end{array}
   $$
 \end{itemize}
\end{definition}

Each element in $A^\sharp$ expresses properties on fields and on an
initialization historic.  This is formalized by the following
correctness relation.
\begin{definition}[Correctness relation]
  $(\may,\must,\wf)$ is a correct approximation of
  $(s,\hist)\in\Static\times\History$, written
  $(\may,\must,\wf)\sim(s,\hist)$ iff:
 $$ \must \subseteq \hist \subseteq \may ~~\mathrm{and}~~\wf \subseteq \{\ f\in\fields \mid s(f)\not=\undef\ \}$$
\end{definition}
This relation expresses that
\begin{enumerate}
\item $\may$ contains all the classes for which we may have called the
  \clinit{} method since the beginning of the program (but it may not
  be finished yet).
\item $\must$ contains all the classes for which we must have called
  the \clinit{} method since the beginning of the program (but it may
  not be finished yet neither).
\item $\wf$ contains all the fields for which we are sure they have
  been written at least once.
\end{enumerate}

The analysis is then specified as a data flow problem.
\begin{definition}[Data flow solution]
  A \emph{Data flow solution} of the Must-Have-Been-Initialized
  analysis is any couple of maps $A_\inA,A_\outA\in\ppt\to A^\sharp$
  that satisfies the set of equations presented in
  Fig.~\ref{fig:analysis}, for all program point $l$ of the program
  $p$.
\end{definition}

\begin{figure}
  $$
  \begin{array}{rcl}
    A_\inA(l) &=& A_0(l) \sqcup A_\first(l) \sqcup \bigsqcup \{A_\outA(l')\mid \intraflow(l',l)\} \\
    A_\outA(l) &=\\
    \multicolumn{3}{r}{
      \makebox[.99\linewidth][r]{$\left\{
          \begin{array}[c]{ll}
            F_\callA\left(F^\init(l,A_\inA(l)),\bigsqcup \{A_\inA(m.\last)\mid \interflow(l,m)\}\right) & \mathrm{if}~ \instrAt(l) = \invoke \\
            F_{\instrAt(l)} (F^\init(l,A_\inA(l))) & \mathrm{otherwise}
          \end{array}\right.$}}
  \end{array}
  $$
  where
  $$
  \begin{array}{rcl}
    A_0(l) &=&     \left\{
      \begin{array}{ll}
        (\emptyset,\emptyset,\emptyset) &\mathrm{if}~l=m_0.\first\\
        \bot & \mathrm{otherwise}\\
      \end{array}\right. \\
    A_\first(l) &=&  \left\{
      \begin{array}{ll}
        \bigsqcup \left\{ F^\init_\callA (C,A_\inA(l')) \mid \clinitflow(l')=C\right\} &\mathrm{if}~l=C.\clinit.\first\\
        \bigsqcup \left\{ F^\init(l',A_\inA(l')) \mid \interflow(l',m)\right\} &\mathrm{if}~l=m.\first\\
        \bot & \mathrm{otherwise}
      \end{array}\right. \\
  \end{array}
  $$
  and $F_{\return}$, $F_{\any}$, $F_{\putS(f)}\in A^\sharp\to A^\sharp$, $F_\callA\in A^\sharp\times A^\sharp\to A^\sharp$,
  $F^\init_\callA\in \class\times A^\sharp\to A^\sharp$ and $F^\init\in \ppt\times A^\sharp\to A^\sharp$ are transfer functions defined by:
  \begin{eqnarray*}
    F_{\return} (\may,\must,\wf) &=& 
    F_{\any} (\may,\must,\wf) = (\may,\must,\wf) \\
    F_{\putS(f)} (\may,\must,\wf) &=& (\may,\must,\wf\cup\{f\}) \\
  \end{eqnarray*}\vspace*{-5ex}
  $$
  \begin{array}{l}
    F_\callA\left((\may_1,\must_1,\wf_1), (\may_2,\must_2,\wf_2)\right) =\\
    ~~~~~~~~~~~~~~~~~~~~~~~~~~~~~~~~~ (\may_2,\must_1\cup\must_2,\wf_1\cup\wf_2)
  \end{array}
  $$
  \begin{eqnarray*}
    F^\init_\callA(C,(\may,\must,\wf)) &=&
    \left\{
      \begin{array}{ll}
        \bot & \mathrm{if}~ C\in\must \\
        (\may\cup\{C\},\must\cup\{C\},\wf) & \mathrm{otherwise}
      \end{array}\right. \\
  \end{eqnarray*}
  $$
  \begin{array}{M{.99\linewidth}}
    F^\init(l,a) = F^\init(l,(\may,\must,\wf)) =\\
    \multicolumn{1}{r}{\left\{
        \begin{array}{ll}
          F_\callA(a, A_\inA(C.\clinit.\last)) & \mathrm{if}~\clinitflow(l)=C~\mathrm{and}~C\not\in\may \\
          F^\init_\callA(C,a) \sqcup F_\callA(a, A_\inA(C.\clinit.\last))\ \ & \mathrm{if}~\clinitflow(l)=C\\[-.5ex]
          & \multicolumn{1}{r}{\mathrm{~and~}C\in\may~\mathrm{and}~C\not\in\must}\\[-.5ex]
          a & \mathrm{otherwise}
        \end{array}\right.} \\
  \end{array}
  $$
  \caption{Data flow analysis}
  \label{fig:analysis}
\end{figure}

In this equation system, $A_\inA(l)$ is the abstract union of three
kinds of data flow information:
(i) $A_0(l)$ gives the abstraction of the initial states if $l$ is the
starting point of the main method $m_0$;
(ii) $A_\first(l)$ is the abstract union of all calling contexts that
may be transferred to $l$ if it is the starting point of a method
$m$.  We distinguish two cases, depending on whether $m$ is the class initializer
of a class $C$.  If it is, incoming calling contexts
are transformed with $F_\callA^\init$ which filters unfeasible calling
edges with $\must$ and adds $C$ to $\may$ and $\must$ otherwise. 
Otherwise, incoming calling contexts are transformed with
$F^\init$ (described below) which take into account the potential
class initialization that may have been performed before the call.
(iii) At last, we merge all incoming data flows from predecessors in
the intra-procedural graph.

The equation on $A_\outA(l)$ distinguishes two cases, depending on
$\instrAt(l)$ is a method call or not.  If it is, we merge all data
flows from the end of the potentially called methods and combine them,
using $F_\callA$ described below, with the data flows facts
$F^\init(l,A_\inA(l))$ that is true just before the call.  Otherwise,
we transfer the data flow $A_\inA(l)$, found at entry of the current
instruction with $F^\init$ and then $F_{\instrAt(l)}$.  While $F^\init$
handles potential class initialization that may have been performed
before the instruction, $F_{\instrAt(l)}$ simply handles the effect of
the instruction $\instrAt(l)$ in a straightforward manner.

The transfer function $F^\init$ is defined with tree distinct cases.
(i) In the first case, we are sure the class initializer $C.\clinit$
will be called because it has never been called before.  We can hence
use safely the last data flow of $C.\clinit$ but we combine it with
$a$ using the operator $F_\callA$ described below.
(ii) In the second case, we are sure that no class initializer will be
called, either because there is no initialization edge at all, or
there is one for a class $C$ but we know that $C.\clinit$ has already
been called.
(iii) In the last case, the two previous cases may happen so we merge
the corresponding data flows.

At last, $F_\callA$ is an operator which combines dataflows about
calling contexts and calling
returns. 
It allows to recover some \emph{must} information that may have been
discarded during the method call because of spurious calling
contexts.  It is based on the monotony of $\must$ and $\wf$: these sets
are under-approximations of initialization history and initialized
fields but since such sets only increase during execution a correct
under-approximation $(\must,\wf)$ at a point $l$ is still a correct
approximation at every point reachable from $l$.

\begin{figure}
  \centering
  \begin{minipage}[b]{6.8cm}
    \pgfimage[width=6.8cm]{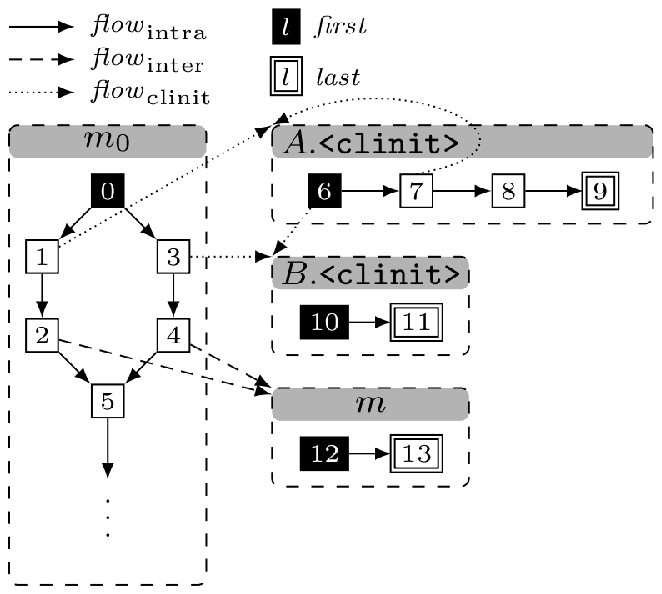}
    \begin{center} \footnotesize
      (a) A program example
    \end{center}
  \end{minipage}
  \begin{minipage}[b]{6.7cm}\tiny
    $$
    \begin{array}[t]{|c|l|ccc|ccc|}\hline
      \multirow{2}{*}{$l$} & \multirow{2}{*}{$\instrAt(l)$} & \multicolumn{3}{c|}{A_\inA(l)} & \multicolumn{3}{c|}{A_\outA(l)} \\
      & & \may & \must & \wf & \may & \must & \wf \\\hline
      0 & \any    & \emptyset & \emptyset& \emptyset& \emptyset & \emptyset & \emptyset \\
      1 & \any    & \emptyset & \emptyset & \emptyset & \{A,B\} & \{A,B\} & \emptyset\\
      2 & \invoke & \{A,B\} & \{A,B\} & \{f\} & \{A,B\} & \{A,B\} & \{f\}\\
      3 & \any    & \emptyset & \emptyset & \emptyset & \{A,B\} & \{B\} & \emptyset\\
      4 & \invoke & \{A,B\} & \{B\} & \emptyset & \{A,B\} & \{B\} & \emptyset\\
      5 & \any    & \{A,B\} & \{B\} & \emptyset & \{A,B\} & \{B\} & \emptyset\\\hline
      6 & \putS(f)& \{A\} & \{A\} & \emptyset & \{A,B\} & \{A,B\} & \{f\}\\
      7 & \any    & \{A,B\} & \{A,B\} &\{f\}& \{A,B\} & \{A,B\} & \{f\}\\
      8 & \return & \{A,B\} & \{A,B\} &\{f\}& \{A,B\} & \{A,B\} & \{f\}\\
      9 &         & \{A,B\} & \{A,B\} & \{f\}& \{A,B\} & \{A,B\}& \{f\} \\\hline
      10& \return & \{A,B\} & \{B\} & \emptyset & \{A,B\} & \{B\} & \emptyset \\
      11&         & \{A,B\} & \{B\} & \emptyset & \{A,B\} & \{B\} & \emptyset\\\hline
      12& \return & \{A,B\} & \{B\}& \emptyset& \{A,B\} & \{B\} & \emptyset\\
      13&         & \{A,B\} & \{B\}& \emptyset& \{A,B\} & \{B\} & \emptyset\\\hline
    \end{array}
    $$
    \begin{center} \footnotesize
      (b) Least dataflow solution of the program example
    \end{center}
  \end{minipage}
  \caption{Program and analysis example}
  \label{fig:example}
\end{figure}

The program example presented is Fig.~\ref{fig:example} is given with
the least solution of its corresponding dataflow problem.  In this
example, $A.\clinit$ has two potential callers in $1$ and $7$ but we
don't propagate the dataflow facts from $7$ to $6$ because we know
that $A.\clinit$ has already been called at this point, thanks to
$\must$.  At point $2$, the method $m$ is called but we don't propagate
in $A_\outA(2)$ the exact values found in $A_\inA(m.\last)$ because we
would lose the fact that $A\in\must$ before the call.  That is why we
combine $A_\inA(2)$ and $A_\inA(m.\last)$ with $F_\callA$ in order to
refine the must information $\must$ and $\wf$.



\begin{theorem}[Computability]
  The least data flow solution for the partial order $\sqsubseteq$ is
  computable by the standard fixpoint iteration techniques.
\end{theorem}
\begin{proof}
  This is consequence of the facts that each equation is monotone,
  there is a finite number of program points in $p$ and
  $(A^\sharp,\sqsubseteq,\sqcup,\sqcap)$ is a finite lattice.
\end{proof}

\begin{theorem}[Soundness]
  If $(A_\inA,A_\outA)$ is a data flow solution then for all reachable
  states $\st{i,\cs,s,\hist}\in\ia{p}$, $A_\inA(i)\sim(s,\hist)$
  holds.
\end{theorem}
\def\proofname{Proof sketch.}
\begin{proof}
  We first define an intermediate semantics $\leadsto$ which is shown
  equivalent to the small-step relation $\to$ but in which method
  calls are big-steps: for each point $l$ where a method $m$ is called
  we go in one step to the intra-procedural successor of $l$ using the
  result of the transitive closure of $\leadsto$.  Such a
  semi-big-step semantics is easier to reason with method calls.
  Once $\leadsto$ is define, we prove a standard subject reduction
  lemma between $\leadsto$ and $\sim$ and we conclude.
\end{proof}
\def\proofname{Proof.}

\section{Handling the Full Bytecode}
\label{sec:extensions}




\subsection{Exceptions}
\label{sec:exceptions}

From the point of view of this analysis, exceptions only change the
control flow graph.  As the control flow graph is computed separately,
it should not change the analysis herein described.

However, if we really need to be conservative, loads of instructions
may throw runtime exceptions (\texttt{IndexOutOfBoundException},
\texttt{NullPointerException}, etc.) or even errors
(\texttt{OutOfMemoryError}, etc.) and so there will be edges in the
control flow graph from most program points to the exit point of the
methods, making the analysis very imprecise.

There are several ways to improve the precision while safely handling
exceptions.  First, we can prove the absence of exception for some of
those (\emph{e.g.} see~\cite{hubert08-2:nonnull_annotation_inferencer}
to remove most \texttt{NullPointerException}s and~\cite{bodik00abcd}
to remove the \texttt{IndexOutOfBoundException}s you need to remove).
Then it is cheap to analyse, for each method, the context in which the
method is called, \emph{i.e.} the exceptions that may be caught if
they are thrown by the method: if there are no handler for some
exception in the context of a particular method, then there is no use
to take in account this exception in the control flow graph of this
method.  Indeed, if such an exception were thrown it would mean the
termination of the program execution so not taking in account the
exception may only add potential behaviours, which is safe.

\subsection{Inheritance}
\label{sec:class-hierarchy}
In the presence of a class hierarchy, the initialization of a class
starts by the initialization of its superclass if it has not been done
yet.  There is therefore an implicit edge in the control flow graph
from each \clinit{} method to the \clinit{} method of its superclass
(except for \texttt{Object.\clinit}).
Although it does not involve any challenging problem, the semantics
and the formalization need to be modified to introduce a new label at
the beginning of each \clinit{} method such that, if $l$ is the label
we introduce at the beginning of \texttt{C.\clinit}, then
$\clinitflow(l)=\textit{super}(C)$.

Note that it is not required to initialize the interfaces a class
implements, nor the super interfaces an interface extends
(cf.~Sect.~2.17.4 of~\cite{lindholm99:jvm_spec}).




\subsection{Field Initializers and Initialization order}
\label{sec:initialization-order}

Although the official JVM specification states (Sect.~2.17.4) that the
initialization of the superclass should be done before the
initialization of the current class and that the field initializers
are part of the initialization process, in Sun's JVM the field
initializers are used to set the corresponding fields before starting
the initialization of the superclass.  This changes the semantics
but removes potential defects as this way it is impossible for some
code to read the field before it has been set.

If the analysis is targeted to a such JVM implementation, depending on
the application of the analysis, the fields which have a field
initializer can be safely ignored when displaying the warnings found,
or be added to \wf{} either when the corresponding class is added to
$\must$ or at the very beginning ($m_o.\first$).  In order for the
analysis to be compatible with the official specification, the
analysis needs to simulate the initialization of the fields that have
a field initializer at the beginning of the class initializer, just
after the implicit call to the class initializer of its superclass.

\subsection{Reflection, User-Defined Class-Loader, Class-Path
  Modification, etc.}
\label{sec:reflection}
There are several contexts in which this analysis can be used: it can
be used to find bugs or to prove the correctness of some code, either
off-line or in a PCC~\cite{nec97} architecture.

In case it used to find bugs, the user mainly need to be aware that
those features are not supported.
\\
If it is used to prove at compile time the correctness of some code,
the analysis needs to handle those features.  In case of reflection,
user-defined class-loaders or modification of the class-path, it is
difficult to be sure that all the code that may be executed has been
analyzed.  The solution would restrict the features of the language in
order to be able to infer what code \emph{may} be executed, which is
an over-approximation of the code that \emph{will} be executed, and to
analyze this code.  For example, Livshits \emph{et al.} proposed
in~\cite{livshits05} an analysis to correctly handle reflection.

In a PCC architecture, the code is annotated at compile time and
checked at run time.  The issue is no more to find the code that will
be executed, because the checking is done at run time when it is a lot
easier to know what code will be executed, but to consider
\emph{enough} code when annotating.  This can be done the same way as
with off-line proofs and by asking the programmer when it is not
possible to infer a precise enough solution.  In this case, if the
user gives incorrect data the checker will notice it while checking
the proof at run time.

\section{Two Possible Applications for the Analysis}
\label{sec:two-possible-uses}

\subsection{Checking That Fields are Written Before Being Read}
\label{sec:fields-are-written-then-read}

The analysis presented in this paper computes a set of fields that
have been written for each program point.  To check that all fields
are written before being read, \emph{i.e.} that when a field
\texttt{C.f} is read at program point $l$, we need to check that
\texttt{C.f} is in the set of written field at this particular program
point.

\subsection{Nullness Analysis of Static Fields}
\label{sec:nullness-analysis}

While in our previous
work~\cite{hubert08-1:nonnull_annotations_inference} we choose to
assume no information about static fields, several tools have targeted
the analysis of static fields as part of their analysis but missed the
issue of the initialization herein discussed.

To safely handle static fields while improving the precision, we can
abstract every field by the abstraction of the values that may be
written to it and, if the static field may be read before being
initialized, then we add the abstraction of the \texttt{null} constant
to the abstraction of the field.  It is a straightforward extension of
the analysis presented
in~\cite{hubert08-1:nonnull_annotations_inference}.

\section{Related work}
\label{sec:related-work}

Kozen and Stillerman studied
in~\cite{kozen02:_eager_class_initial_for_java} eager class
initialization for Java bytecode and proposed a static analysis based
on fine grained circular dependencies to find an initialization order
for classes.  If their analysis finds an initialization order, then
our analysis will be able to prove all fields are initialized before
being read.  If their analysis finds a circular dependency, it fails
to find an initialization order and issues an error while our analysis
considers the initialization order implied by the main program and may
prove that all fields are written before being read.

Instance field initialization have been studied for different
purposes.  Some works are focused on null-ability properties such as
Fähndrich and Leino in~\cite{fahndrich03:_declar_and_check_non_null},
our work in~\cite{hubert08-1:nonnull_annotations_inference} or
Fähndrich and Xia
in~\cite{fahndrich07:object_invariants_delayed_types}.  Other work
have been focused on different properties such as Unkel and
Lam~\cite{unkel08:infererence_stationary_fields} who studied
stationary fields.  Instance field initialization offers different
challenges from the one of static fields: the initialization method is
explicitly called soon after the object allocation.

Several formalizations of the Java bytecode have been proposed that,
among other features, handled class initialization such as the work of
Debbabi \emph{et al.} in~\cite{FormalJVM} or Belblidia and Debbabi
in~\cite{FormalJVM2}.  Their work is focused on the dynamic semantics
of the Java bytecode while our work is focused on its analysis.

Böerger and Schulte~\cite{Borger98aprogrammer} propose another dynamic
semantics of Java.  They consider a subset of Java including
initialization, exceptions and threads.  They have
exhibited~\cite{borger00initialization} some weaknesses in the
initialization process as far as the threads are used.  They pointed
out that deadlocks could occur in such a situation.

Harrold and Soffa~\cite{harrold94:interprocedural_def_use} propose an
analysis to compute inter-procedural definition-use chains.  They have
not targeted the Java bytecode language and therefore neither the
class initialization problems we have faced but then, our analysis can
be seen as a lightweight inter-procedural definition-use analysis
where all definitions except the default one are merged.

Hirzel \emph{et
  al.}~\cite{hirzel04:pointer_analysis_dynamic_class_loading} propose
a pointer analysis that target dynamic class loading and lazy class
initialization.  There approach is to analyse the program at run time,
when the actual classes have been loaded and to update the data when a
new class is loaded and initialized.  Although it is not practical to
statically certify programs, a similar approach could certainly be
adapted to implement a checker in PCC architecture such as the one
evoked in Sect.~\ref{sec:reflection}.

\section{Conclusion and Future Work}
\label{sec:conclusion}
We have shown that class initialization is a complex mechanism and
that, although in most cases it works as excepted, in some more
complicated examples it can be complex to understand in which
order the code will be executed.  More specifically, some fields may
be read before being initialized, despite being initialized in their
corresponding class initialization methods.
A sound analysis may need to address this problem to infer precise and
correct information about the content of static fields.
We have proposed an analysis to identify the static fields that may be
read before being initialized and shown how this analysis can be used
to infer more precise information about static fields in a sound
null-pointer analysis.

We expect the analysis to be very precise if the control flow graph is
\emph{accurate enough}, but we would need to implement this analysis
to evaluate the precision needed for the control flow graph.

\bibliographystyle{entcs} \bibliography{biblio}
\end{document}